\newif\ifpaper
\paperfalse

\ifpaper
\documentclass[runningheads,envcountsect,envcountreset,envcountsame,orivec]{llncs}

\else
\documentclass[a4paper,10pt]{article}

\fi

\usepackage{amsmath}
\usepackage{amsfonts}
\usepackage{amssymb}         
\usepackage{amsopn}
\usepackage{theorem}          
\usepackage{latexsym}
\usepackage{graphicx}        
\usepackage{mathrsfs}		
\usepackage{psfrag}
\usepackage{algorithmic}
\usepackage{algorithm}
\usepackage{color}
\usepackage{verbatim}
\usepackage{url}
\usepackage{tikz}
\usepackage[caption=false,format=hang]{subfig}
\usepackage{enumitem}

\colorlet{darkred}{red!70!black}
\colorlet{darkblue}{blue!70!black}
\colorlet{darkgreen}{green!60!black}

\usetikzlibrary{decorations}
\usetikzlibrary{arrows}

\ifpaper
\else
\newcommand{\qed}{}
\newtheorem{result}{\ }[section]
\theoremstyle{changebreak}                
\newtheorem{theorem}[result]{Theorem}

\newtheorem{lemma}[result]{Lemma}
\newtheorem{corollary}[result]{Corollary}
\newtheorem{proposition}[result]{Proposition}
\newtheorem{example}[result]{Example}
\newenvironment{proof}
 {{\sl Proof.}\hspace*{1 ex}}%
 {{\nopagebreak\hspace*{\fill}$\Box$\par\vspace{12pt}}}
\fi

\newcommand{\transpose}[1]{{#1}^{\top}}
\newcommand{\natop}[2]{\genfrac{}{}{0pt}{}{#1}{#2}}
\begin{document}

\ifpaper

\title{On the number of solutions of the discretizable molecular distance geometry problem}
\titlerunning{On the number of solutions of the DMDGP}
\author{\sc Leo Liberti${}^1$ \and Beno\^{\i}t Masson${}^2$ \and Jon Lee ${}^3$ \and Carlile Lavor ${}^4$ \and Antonio Mucherino${}^5$} 
\institute{LIX, \'Ecole Polytechnique, 91128 Palaiseau, France
  \email{\url{liberti@lix.polytechnique.fr}}\and
  IRISA, INRIA, Campus de Beaulieu, 35042 Rennes, France
  \email{\url{benoit.masson@inria.fr}}\and
  Dept.~of Mathematical Sciences, IBM T.J.~Watson Research Center, PO
  Box 218, Yorktown Heights, NY 10598, USA,
  \email{\url{jonlee@us.ibm.com}}\and
  Department of Applied Mathematics (IMECC-UNICAMP),
  State University of Campinas, C.P.~6065, 13081-970, Campinas - SP,
  Brazil, \email{\url{clavor@ime.unicamp.br}}\and
  INRIA Futurs, Lille, France 
  \email{\url{antonio.mucherino@inria.fr}}}
\authorrunning{Liberti {\it et al.}}

\maketitle

\else

\thispagestyle{empty}
\begin{center} 
{\LARGE On the number of solutions of the discretizable molecular distance geometry problem}
\par \bigskip
{\sc Leo Liberti${}^1$, Beno\^{\i}t Masson${}^2$, Jon Lee ${}^3$, Carlile Lavor ${}^4$, Antonio Mucherino${}^5$}
\par \bigskip
\begin{minipage}{15cm}
\begin{flushleft}
{\small
\begin{itemize}
\item[${}^1$] {\it LIX, \'Ecole Polytechnique, 91128 Palaiseau, France}\\
  Email:\url{liberti@lix.polytechnique.fr} 
\item[${}^2$] {\it IRISA, INRIA, Campus de Beaulieu, 35042 Rennes, France}\\
  Email:\url{benoit.masson@inria.fr}
\item[${}^3$] {\it Dept.~of Mathematical Sciences, IBM T.J.~Watson
  Research Center, PO Box 218, Yorktown Heights, NY 10598, USA}\\
  Email:\url{jonlee@us.ibm.com}
\item[${}^4$] {\it Department of Applied Mathematics (IMECC-UNICAMP),
  State University of Campinas, C.P.~6065, 13081-970, Campinas - SP,
  Brazil} \\ Email:\url{clavor@ime.unicamp.br}
\item[${}^5$] {\it INRIA Futurs, Lille, France} \\
  Email:\url{antonio.mucherino@inria.fr}
\end{itemize}
}
\end{flushleft}
\end{minipage}
\par \medskip \today
\end{center}
\par \bigskip

\fi

\begin{abstract} 
The Generalized Discretizable Molecular Distance Geometry Problem is a
distance geometry problems that can be solved by a combinatorial
algorithm called ``Branch-and-Prune''. It was observed empirically
that the number of solutions of YES instances is always a power of
two. We give a proof that this event happens with probability one.
\end{abstract}


\section{Introduction}
We consider the following problem arising in the analysis of Nuclear
Magnetic Resonance (NMR) data for general molecules.
\begin{quote}
{\sc Molecular Distance Geometry Problem} (MDGP).\par
 Given an undirected graph $G=(V,E)$ and a function $d:E\to\mathbb{R}$,
decide whether there is an embedding $x:V\to\mathbb{R}^3$ such that
\begin{equation}
  \forall \{u,v\}\in E \quad (||x_u-x_v||=d_{uv}) \label{mdgpeq}
\end{equation}
\end{quote}
The MDGP is usually cast as a nonconvex unconstrained mathematical
program $\min \sum_{\{u,v\}\in E} (||x_u-x_v||^2 - d_{uv}^2)^2$, which
can be solved using continuous Global Optimization techniques
\cite{lln1}. See \cite{lln4,mdgpsurvey} for surveys.

If the NMR data is obtained from a protein, it is usually possible to
distinguish between backbone atoms and atoms of side chains. We can
then attempt to place the backbone first \cite{lln5} and the side
chains later \cite{santana}. We therefore assume that $G$ is the graph
of the backbone. If a total order is given on $V$ such that for each
atom $v$ there are at least three atoms that are adjacent to it in $E$
and preceding it in the order (these are called the {\it adjacent
  predecessors} of $v$), if strict triangle inequality holds for $d$
restricted to at least one triplet of adjacent predecessor for each
vertex, and if spatial positions are given for the first three atoms,
there is only a finite number of possible embeddings compatible with
the given distances, which can all be found using the so-called
Branch-and-Prune (BP) algorithm \cite{lln5}. The subset of MDGP
instances for which the described atomic order exists is called {\sc
  Discretizable MDGP} (DMDGP) \cite{mdgpsurvey} and is known to be
{\bf NP}-hard \cite{lln2}. It was empirically observed that for any
given practical instance, BP always finds a number of solution that is
a power of two \cite{lln5}. An isolated earlier counterexample,
derived via a complexity reduction, given as Lemma 5.1 in \cite{lln2},
falsifies this conjecture. In this paper we give a formal description
of the DMDGP in arbitrary dimensions (Sect.~\ref{s:ddgp}), of the BP
algorithm and some of its theoretical properties
(Sect.~\ref{s:bp}). We then study some geometrical aspects of the BP
tree (Sect.~\ref{s:symm}), and prove that the number of solutions of
feasible DMDGP instances is a power of two with probability one
(Sect.~\ref{s:nbsol}). We also exhibit a family of counterexamples
(Sect.~\ref{s:counterex}) that are much more intuitive than the one
given in \cite{lln2}.

\section{The Discretizable Molecular Distance Geometry Problem}
\label{s:ddgp}
The generalization of the MDGP to arbitrary dimensions asks for an
embedding of $G$ in $\mathbb{R}^K$ satisfying \eqref{mdgpeq} and is
called the {\sc Distance Geometry Problem} (DGP). The generalization
of the DMDGP to $\mathbb{R}^K$ replaces triplets with $K$-uples of
adjacent predecessors and strict triangle with strict simplex
inequalities \cite{jiao}. For a set $U=\{x_i\in\mathbb{R}^K\;|\;i\le
K+1\}$ of points in $\mathbb{R}^{K}$, let $D$ be the symmetric matrix
whose $(i,j)$-th component is $\|x_i-x_j\|^2$ for all $i,j\le K+1$ and
let ${D'}$ be $D$ bordered by a left $\transpose{(0,1,\ldots,1)}$
column and a top $(0,1,\ldots,1)$ row (both of size $K+2$). Then the
Cayley-Menger formula states that the volume $\Delta_K(U)$ of the
$K$-simplex on $U$ is given by $\Delta_K(U)
=\sqrt{\frac{(-1)^{K+1}}{2^K(K!)^2}|{D'}|}$.  The strict simplex
inequalities are given by $\Delta_K(U)>0$. For $K=3$, these reduce to
strict triangle inequalities. We remark that only the distances of the
simplex edges are necessary to compute $\Delta_K(U)$, rather than the
actual points in $U$; the needed information can be encoded as a
complete graph ${\bf K}_{K+1}$ on $K+1$ vertices with edge weights as
the distances.

Let $n=|V|$ and $m=|E|$. For all $v\in V$, let $N(v)=\{u\in
V\;|\;\{u,v\}\in E\}$ be the star of vertices around $v$ (also called
the adjacencies of $v$); for a directed graphs $(V,A)$, where
$A\subseteq V\times V$, we denote the outgoing star by $N^+(v)=\{u\in
V\;|\;(v,u)\in A\}$. For an order $<$ on $V$, let $\gamma(v)=\{u\in
V\;|\;u<v\}$ be the set of predecessors of $v$, and let
$\rho(v)=|\gamma(v)|+1$ the rank of $v$ in $<$. For $V'\subseteq V$,
we denote by $G[V']$ the subgraph of $G$ induced by $V'$. For a finite
set $M$, let $\mathcal{P}(M)$ be its power set. We call an embedding
$x$ of $G$ {\it valid} if (\ref{mdgpeq}) holds for $G$. For a sequence
$x=(x_1,\ldots,x_n)$ and a subset $U\subseteq\{1,\ldots,n\}$ we let
$x[U]$ be the subsequence of $x$ indexed by $U$. If $x$ is an initial
subsequence of $y$, then $y$ is an {\it extension} of $x$.

\begin{quote}
{\sc Generalized DMDGP} (GDMDGP). Given an undirected graph $G=(V,E)$,
an edge weight function $d:E\to\mathbb{R}_+$, an integer $K>0$, a
subset $V_0\subseteq V$ with $|V_0|=K$, a partial embedding
$\bar{x}:V_0\to\mathbb{R}^K$ valid for $G[V_0]$, and a total order $<$
on $V$ such that:
\begin{eqnarray}
   \{v\in V\;|\;\rho(v)\le K\}&=&V_0; \label{ddgpeq0} \\
  \forall v\in V \quad (\rho(v)>K \to |N(v)\cap\gamma(v)| &\ge&
    K); \label{ddgpeq1} \\ 
  \forall v\in V\smallsetminus V_0\; \exists U_v\subseteq
    N(v)\cap\gamma(v) \; (G[U_v]&=& {\bf K}_K \land \nonumber \\
     \Delta_{K-1}(U_v)>0\land\forall u\in U_v (\rho(v)-K\le \rho(u) 
     \le \rho(v) &-& 1)), \label{ddgpeq2}
\end{eqnarray}
decide whether there is a valid extension $x:V\to\mathbb{R}^K$ of
$\bar{x}$.
\end{quote}
Conditions (\ref{ddgpeq0}-\ref{ddgpeq2}) allow the search for vertex
$v$ to only depend on the $K$ vertices of rank preceding $\rho(v)$, as
$x_v$ is the intersection of at least $K$ spheres centered at $x_u$
and with radius $d_{uv}$ for all $u\in U_v$. This, in particular,
implies that the predecessors of $v$ are placed before $v$, so that
all of the distances between all predecessors are known when placing
$v$. Thus, we can replace the condition $G[U_v]={\bf K}_K$ in
\eqref{ddgpeq2} by $|U_v|=K$. We remark that eliminating the condition
$\forall u\in U_v (\rho(v)-K\le \rho(u) \le \rho(v)-1)$ from
\eqref{ddgpeq2} and fixing $K$ gives rise to a problem called the {\sc
  Discretizable Distance Geometry Problem} in $K$ dimensions
(DDGP${}_K$), discussed in \cite{ddgp}, which is itself a subset of
instances of the DDGP, where the number of dimensions is not
fixed. The DMDGP is therefore the subset of the GDMDGP for which
$K=3$ \cite{lln2}. In summary, we have the following diagram (where
arrows represent the $\subset$ relation).
\begin{center}
\psfrag{DMDGP}{\scriptsize DMDGP}
\psfrag{DDGP3}{\scriptsize DDGP${}_3$}
\psfrag{MDGP}{\scriptsize MDGP}
\psfrag{DGP}{\scriptsize DGP}
\psfrag{GDMDGP}{\scriptsize GDMDGP}
\psfrag{DDGP}{\scriptsize DDGP}
\includegraphics[width=9cm]{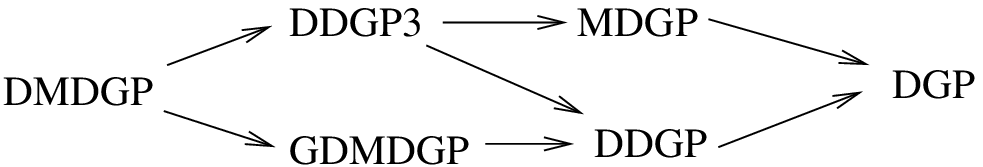}
\end{center}
A polynomial reduction from {\sc Subset-Sum} to the DMDGP shows that
the GDMDGP is {\bf NP}-hard. Even restricting the range of $d$ to the
set of rational numbers, it is not clear whether there are YES
certificates for the DDGP that have polynomially bounded size.

In the following, we assume that the probability of any point of
$\mathbb{R}^K$ belonging to any given subset of $\mathbb{R}^K$ having
Lebesgue measure zero is equal to zero.

\section{Branch-and-Prune}
\label{s:bp}
The BP algorithm for the DMDGP, presented in \cite{lln5}, can easily
be generalized to the GDMDGP. As mentioned above, once the vertices of
$U_v$ have been embedded in $\mathbb{R}^K$, the known distances
from vertices in $U_v$ to a given $v$ will enforce the position of $v$
as the intersection of $K$ spheres. Under strict simplex inequalities,
this intersection consists of at most two distinct points. The BP
exploits this fact to recursively generate a binary search tree of
height at most $n$ where a node at level $i$ represents a possible
placement in $\mathbb{R}^K$ of the vertex of $G$ with rank $i$ in
$<$. Paths of length $n$ correspond to valid embeddings.

Let $G$ be a GDMDGP instance. Consider $v\in V$ with rank
$\rho(v)=i>K$, let $G^v=G[\gamma(v)\cup\{v\}]$ and $x$ be a valid
embedding of $G[\gamma(v)]$. We characterize the number of extensions
of $x$ valid for $G^v$ in the following lemmata (which also hold
for the DDGP); the proof technique for Lemma \ref{atmosttwolemma} is
well known \cite{dongwu,protti}; we repeat the proof here because we
believe the explicit form \eqref{quadraticzK} is unpublished.
\begin{lemma}
If $|N(v)\cap\gamma(v)|=K$ then there are at most two distinct
extensions of $x$ that are valid for $G^v$. If one valid extension
exists, then with probability 1 there are exactly two distinct valid
extensions.
 \label{atmosttwolemma}
\end{lemma}
\begin{proof}
Since $|N(v)\cap\gamma(v)|=K$, $U_v=N(v)\cap\gamma(v)$ and
$v$ is at the intersection of exactly $K$ spheres in $\mathbb{R}^K$
(each centered at $x_u$ with radius $d_{uv}$, where $u\in U_v$). The
position $z\in\mathbb{R}^K$ of $v$ must then satisfy:
\begin{equation}
  \forall u\in U_v \quad \|z-x_u\| = d_{uv} \Rightarrow
   \|z\|^2-2 x_u\cdot z+\|x_u\|^2 = d_{uv}^2. \label{prelin}
\end{equation}
As in \cite{dong03}, we choose an arbitrary $w\in U_v$, say
$w=\max_<U_v$, and subtract from the Eq.~(\ref{prelin}) indexed by $w$
the other equations of (\ref{prelin}), obtaining the system:
\begin{equation}
  \left.\begin{array}{rcl}
  \forall u\in U_v \smallsetminus \{w\} \quad
    2(x_u-x_w)\cdot z &=& (\|x_u\|^2-d_{uv}^2) - (\|x_w\|^2-d_{wv}^2) \\
  \|z\|^2-2x_w\cdot z+\|x_w\|^2 &=&
  d_{wv}^2. \end{array}\right\} \label{quadsys} 
\end{equation}
The system (\ref{quadsys}) consists of a set of $K-1$ linear equations
and a single quadratic equation in the $K$-vector $z$. We write the
linear equations as the system $Az=b$, where the $(u,j)$-th component
of $A$ is $2(x_{uj}-x_{wj})$, the $u$-th component of $b$ is
$\|x_u\|^2-\|x_w\|^2-d_{uv}^2+d_{wv}^2$, $A$ is $(K-1)\times K$ and
$b\in\mathbb{R}^{K-1}$. By strict simplex inequality, $A$ has full
rank (for otherwise $\sum_{u\not=w} \lambda_u(x_u-x_w)=0$ implies that
$x_w$ is in the span of $\{x_u\;|\;u\in U_v\}$, and hence that
$\Delta_{K-1}(U_v)=0$); so without loss of generality assume that the
square matrix $B$ formed by the first $K-1$ columns of $A$ is
invertible. Let $z_B$ be the vector consisting of the first $K-1$
components of $z$; then the linear part (first $K-1$ equations) of
(\ref{quadsys}) yields $z_B=B^{-1}(b-Nz_K)$, where
$N=2(x_{uK}-x_{wK}\;|\;u\in
U_v\smallsetminus\{w\})\in\mathbb{R}^{K-1}$. After replacement of
$z_B$ in (\ref{quadsys}) with $z_B(z_K)$, we obtain the following
quadratic equation in $z_K$:
\begin{equation}
(\|\bar{N}\|^2+1)z^2_K - 2((\bar{b}+x_{wB})\bar{N}+x_{wK})z_k+
  (\|x_{wB}-\bar{b}\|^2+x_{wK}^2-d_{wv}^2)=0,
\label{quadraticzK}
\end{equation}
where $\bar{b}=B^{-1}b$ and $\bar{N}=B^{-1}N$.  If the discriminant of
\eqref{quadraticzK} is negative then no extension of $\bar{x}$ to $v$
is possible and the result follows. If the discriminant is
nonnegative, \eqref{quadraticzK} has solutions $z'_K,z''_K$ yielding
points $z'=(z_B(z_K'),z_K')$ and
$z''=(z_B(z_K''),z_K'')\in\mathbb{R}^K$, which are distinct with
probability 1 because the discriminant is zero with probability 0.
The extended embeddings, distinct with probability 1, are given by
$(x,z')$ and $(x,z'')$.  \qed
\end{proof}

\begin{lemma}
If $|N(v)\cap\gamma(v)|>K$ then, with probability~$1$, there is at
most one extension of $x$. \label{atmostonelemma}
\end{lemma}
\begin{proof}
Consider a subset $S\subseteq N(v)\cap\gamma(v)$ such that
$|S|=K+1$ and $S\supseteq U_v$. Either there is at least one point
$x_v$ such that $(x,x_v)$ is an embedding of $G[S\cup\{v\}]$ that is
valid w.r.t.~the system:
\begin{equation}
  \forall u\in S \quad \sum_{k\le K}
    (x_{vk}^2-2x_{uk}x_{vk}+x_{uk}^2) = d_{uv}^2, \label{lins1}
\end{equation}
or the system has no solution. In the latter case, the result follows,
so we assume now that there is a point $x_v$ satisfying
(\ref{lins1}). Since the points $x_u$ are known for all $u\in S$,
(\ref{lins1}) is a quadratic system with $K$ variables and $K+1$
equations. As in the proof of Lemma \ref{atmosttwolemma}, we derive an
equivalent linear system from (\ref{lins1}). Since $d$ satisfies the
strict simplex inequalities on $U_v$ with probability~$1$ and
$S\supseteq U_v$, by \cite{coope} $\{x_u\;|\;u\in S\}$ are not
co-planar and the system has exactly one solution.  \qed
\end{proof}

\begin{lemma}
With the notation of Lemma \ref{atmosttwolemma}, if $\bar{x}$ is a
valid embedding for $G[U_v]$, then $z''$ is a reflection of $z'$ with
respect to the hyperplane through the $K$ points of $\bar{x}$.
\label{reflection}
\end{lemma}
\begin{proof}
Any sphere in $\mathbb{R}^K$ is symmetric with respect to any
hyperplane through its center; so the intersection of up to $K$
spheres in $\mathbb{R}^K$ is symmetric with respect to the hyperplane
containing all the centers. \qed
\end{proof}

\label{nota:reflection}Reflections with respect to hyperplanes are isometries, and can
therefore be represented by linear operators. If $a\in\mathbb{R}^K$ is
the unit normal vector to a hyperplane $H$ containing the origin,
then the reflection operator $R_0$ w.r.t.~$H$ can be expressed in
function of the standard basis by the matrix $I-2a\transpose{a}$, where
$I$ is the $K\times K$ identity matrix \cite{brady}. If $H$ is a
hyperplane with equation $\transpose{a} x = a_0$, %
with $a_0\not=0$, s.t.~$x$ is ordered and $a_i$, for some $1 \leq i \leq K$, is the nonzero coefficient of smallest index in $a$.
Then, the reflection operator $R$
acting on a point $p\in\mathbb{R}^K$ w.r.t.~$H$ is given by
$R(p) = R_0(p-\frac{a_0}{a_i}e_i)+\frac{a_0}{a_i}e_i$, where $e_i \in \mathbb{R}^K$ is the unit vector with $1$ at index $i$ and $0$ elsewhere: we first we translate $p$ so that
we can reflect it using $R_0$ w.r.t.~the translation of $H$ containing the
origin, then we perform the inverse translation of the reflection.

A formal description of the BP algorithm for the GDMDGP (which also
holds for the DDGP) is given in Alg.~\ref{alg1}. It builds a binary
search tree $\mathcal{T}=(\mathcal{V},\mathcal{A})$, directed from the
root to the leaves, whose nodes are triplets
$\alpha=(x(\alpha),\lambda(\alpha),\mu(\alpha))$. For
$\alpha\in\mathcal{T}$ we denote by $\mbox{\sf p}(\alpha)$ the unique
path from the root node {\sf r} of $\mathcal{T}$ to $\alpha$;
$x(\alpha)$ is an extension of the embedding $x^-$ found on $\mbox{\sf
  p}(\alpha^-)$, where $\alpha^-$ is the unique parent node of
$\alpha$. Next, $\lambda(\alpha)\in\{0,1\}$ distinguishes whether
$\alpha$ is a ``left'' or a ``right'' subnode of $\alpha^-$. More
precisely, let $\alpha$ be a node at level $i$ in $\mathcal{T}$,
$v=\rho^{-1}(i)$, $\bar{x}$ be a partial embedding of $G[U_v]$, and
$\transpose{a}_v x= a_{v0}$ be the equation of the ($(K-1)$-dimensional by
\eqref{ddgpeq2}) hyperplane through the points of $\bar{x}$. Assuming $u=\rho^{-1}(i-1)$, $a_v\in\mathbb{R}^K$ is oriented so that $a_v\cdot a_u \geq 0$; then:
\begin{equation}
\lambda(\alpha)=\left\{\begin{array}{lrl}
  0 & \mbox{ if } & \transpose{a}_v x(\alpha)_i\le a_{v0} \\
  1 & \mbox{ if } & \transpose{a}_v x(\alpha)_i>a_{v0}.
  \end{array}\right. \label{avorientation}
\end{equation}
Lastly, $\mu(\alpha)=\boxplus$ if $x$ is a valid extension of $x^-$,
in which case the node is said to be {\it feasible}, and
$\mu=\boxminus$ otherwise. This allows us to retrieve the set $X$ of
all valid embeddings of $G$ by simply traversing $\mathcal{T}$
backwards from the leaf nodes marked $\boxplus$ up to {\sf r}.

\begin{algorithm}[!htp]
\begin{algorithmic}[1]
\REQUIRE Partial embedding $\bar{x}$ of first $K$ vertices of $G$
\ENSURE Set $X$ of valid embeddings of $G$
\STATE Let $\alpha=(\bar{x}_1,0,\boxplus)$ and
$\alpha'=(\bar{x}_1,1,\boxminus)$ 
\STATE Initialize
  $\mathcal{V}=\{\alpha,\alpha'\}$ and 
  $\mathcal{A}=\{(\mbox{\sf r},\alpha),(\mbox{\sf r},\alpha')\}$
\FOR{$1<i\le K$}
  \STATE Let $\alpha=(\bar{x}_i,0,\boxplus)$,
  $\alpha'=(\bar{x}_i,1,\boxminus)$,
  $\beta=(\bar{x}_{i-1},0,\boxplus)$
  \STATE Let $\mathcal{V}\leftarrow\mathcal{V}\cup\{\alpha,\alpha'\}$ and 
  $\mathcal{A}\leftarrow\mathcal{A}\cup\{(\beta,\alpha),(\beta,\alpha')\}$
  \label{startbin}
\ENDFOR
\STATE {\sc BranchAndPrune($K+1$, $(\bar{x}_K,0,\boxplus)$)}
\STATE Let $X=\{x(\theta)\;|\;\theta\in\mathcal{V}\land
|N^+(\theta)|=0\land \mu(\theta)=\boxplus\}$ \label{valid1}
\STATE {\bf stop}
\STATE
\STATE {\bf function} {\sc BranchAndPrune($i$, $\beta$)}:
\IF{$i>n \vee \mu=\boxminus$}
  \RETURN
\ENDIF
\STATE Let $v=\rho^{-1}(i)$
\STATE Choose $U_v$ such that Eq.~\eqref{ddgpeq2} holds
\STATE Compute the equation $\transpose{a}_v x = a_{v0}$ of the hyperplane through $x[U_v]$
\STATE Let $Z=\{z',z''\}$ be extensions of $x(\beta)$ to $v$, and
$Z'=Z$ \label{valid2} 
\FOR{$z\in Z$}
  \IF{$\exists \{u,v\}\in E\; \|x(\beta)_u-z\|\not=d_{uv}$} \label{step5}
     \STATE Let $Z=Z\smallsetminus\{z\}$ \label{valid3}
  \ENDIF
\ENDFOR
\IF{$Z=\{z',z''\}$}
  \IF{$\transpose{a}_v z' \leq a_{v0}$}
    \STATE Let $\alpha=(z',0,\boxplus)$,
    $\alpha'=(z'',1,\boxplus)$ \label{chiinj1a} 
  \ELSE
    \STATE Let $\alpha=(z'',0,\boxplus)$,
    $\alpha'=(z',1,\boxplus)$ \label{chiinj1b}
  \ENDIF
\ELSIF{$Z=\{z\}$}
  \IF{$\transpose{a}_v z \leq a_{v0}$}
    \STATE Let $\alpha=(z,0,\boxplus)$,
    $\alpha'=(Z'\smallsetminus\{z\},1,\boxminus)$ \label{chiinj2a}
  \ELSE
    \STATE Let $\alpha=(z,1,\boxplus)$,
    $\alpha'=(Z'\smallsetminus\{z\},0,\boxminus)$ \label{chiinj2b}
  \ENDIF
\ELSE
  \RETURN \label{noinfpairs}
\ENDIF
\STATE Let $\mathcal{V}\leftarrow\mathcal{V}\cup\{\alpha,\alpha'\}$ and 
  $\mathcal{A}\leftarrow\mathcal{A}\cup\{(\beta,\alpha),(\beta,\alpha')\}$
  \label{recursbin}
\FOR{$\theta\in N^+(\beta)$ such that $\mu(\theta)=\boxplus$}
  \STATE {\sc BranchAndPrune}($i+1$, $\theta$)
\ENDFOR
\RETURN
\end{algorithmic}
\caption{The Branch and Prune algorithm.}
\label{alg1}
\end{algorithm}

We remark that Alg.~\ref{alg1} differs from the original BP
formulation \cite{lln5} because it applies to $K$ dimensions and
stores the binary search tree. In the rest of this section we list
some of the elementary properties of the BP algorithm.

\begin{lemma}
At termination of Alg.~\ref{alg1}, $X$ contains all valid embeddings
of $G$ extending $\bar{x}$.%
\label{fact1}
\end{lemma}
\begin{proof}
$Z$ exists with probability 1 by Lemma \ref{atmosttwolemma}. Every
  embedding in $X$ is valid because of Steps \ref{valid2} and
  \ref{step5}-\ref{valid3}. No other valid extension of $\bar{x}$
  exists because of Lemmata
  \ref{atmosttwolemma}-\ref{atmostonelemma}. \qed
\end{proof}

We now partition $\mathcal{V}$ in pairwise disjoint subsets
$\mathcal{V}_1,\ldots,\mathcal{V}_n$ where for all $i\le n$ the set
$\mathcal{V}_i$ contains all the nodes of $\mathcal{V}$ at level $i$
of the tree $\mathcal{T}$.

\begin{proposition}
With probability 1, there is no level $i\le n$ having two distinct
feasible nodes $\beta,\theta\in\mathcal{V}_i$ such that $|\{\alpha\in
N^+(\beta)\;|\;\mu(\alpha)=\boxplus\}|=1$ and $|\{\alpha\in
N^+(\theta)\;|\;\mu(\alpha)=\boxplus\}|=2$.
\label{uniformlevel}
\end{proposition}
\begin{proof}
We show that for all $i\le n$ the event of having two distinct nodes
$\beta,\theta\in\mathcal{V}_i$, with $\rho^{-1}(i)=v$, such that
$\beta$ has one feasible subnode and $\theta$ has two has probability
0. Consider $T_v=N(v)\cap\gamma(v)$: if $|T_v|=K$ then by Lemma
\ref{atmosttwolemma} $\beta$ should have exactly two feasible subnodes
with probability 1; since it only has one, the event $|T_v|=K$ occurs
with probability 0. Since $|T_v|\ge K$ by \eqref{ddgpeq2}, the event
$|T_v|>K$ occurs with probability 1. Thus by Lemma
\ref{atmostonelemma} there is at most one valid embedding extending
the partial embedding at $v$, which means that the two feasible
subnodes of $\theta$ represent the same embedding, an event that
occurs with probability 0.  \qed
\end{proof}

\section{Geometry in the BP Trees}
\label{s:symm}

Here we exhibit some geometrical symmetry properties between the valid
embeddings of a GDMDGP instance. We prove in Thm~\ref{fact5} that for
any valid embedding $y \in X$, if the BP tree branches at level $i$ on
the path to $y$, then the embedding obtained by reflecting all the
nodes after level $i$ is also valid.

To this aim, we need to emphasize those BP branchings which carry on
to feasible leaf nodes along both branches. For $y\in X$ and a vertex
$v\in V\smallsetminus V_0$ we denote $\Upsilon(y,v)$ the following
property:
\begin{quote}
  $\Upsilon(y,v)$: {\sl there are feasible leaf nodes
    $\beta,\beta'\in\mathcal{V}_n$ such that $x(\beta)=y$, 
    $\mbox{\sf p}(\beta)\cap\mathcal{V}_{\rho(v)-1}=\mbox{\sf
      p}(\beta')\cap\mathcal{V}_{\rho(v)-1}$ and $\mbox{\sf
      p}(\beta)\cap\mathcal{V}_{\rho(v)}\not=\mbox{\sf
      p}(\beta')\cap\mathcal{V}_{\rho(v)}$}.
\end{quote}
If $\Upsilon(y,v)$ holds, it is easy to show that $\mbox{\sf
  p}(\beta)\cap\mathcal{V}_{\rho(v)-1}$ contains a single feasible
node with two feasible subnodes. With $\Upsilon(y,v)$ true, we let
$R^v$ be the Euclidean reflection operator with respect to the
hyperplane through $y[U_v]$ (as defined on
p.~\pageref{nota:reflection}). Define
$\tilde{R}^v=I^{\rho(v)-1}\times (R^v)^{n-\rho(v)}$,
i.e.~$\tilde{R}^vy=(y_1,\ldots,y_{i-1},R^vy_i,\ldots, R^vy_n)$. We
remark that for all $i\le \ell\le n$ and for all
$\alpha\in\mathcal{V}_\ell$ the set $\mbox{\sf
  p}(\alpha)\cap\mathcal{V}_i$ has a unique element.

\begin{corollary}
Let $\alpha\in\mathcal{V}_{i-1}$ for some $i>1$, $v=\rho^{-1}(i)$ and
$N^+(\alpha)=\{\eta,\beta\}$ with
$\mu(\eta)=\mu(\beta)=\boxplus$. Then $x(\eta)_v=R^v x(\beta)_v$.
\label{reflectionbranch}
\end{corollary}
\begin{proof}
This is a corollary to Lemma \ref{reflection}.
\end{proof}

\begin{lemma}\label{fact4}
Let $\alpha\in\mathcal{V}_{i-1}$ for some $i>1$ such that $N^+(\alpha) = \{\eta',\beta'\}$, $u=\rho^{-1}(i)$; $v>u$ with $\rho(v)=\ell$, and consider two feasible nodes $\eta,\beta\in\mathcal{V}_\ell$ such that $\eta' = \mathsf{p}(\eta) \cap \mathcal{V}_{i}$ and $\beta' = \mathsf{p}(\beta) \cap \mathcal{V}_{i}$.
Then, with probability~$1$, the following statements are equivalent:
  \begin{enumerate}[label=\textnormal{(\emph{\roman*})},leftmargin=*,widest=iii]
  \item\label{it:fact4-1} $\forall\;i \leq j \leq \ell$, $x(\beta'')_w=R^u x(\eta'')_w$, where $\eta'' = \mathsf{p}(\eta) \cap \mathcal{V}_{j}$, $\beta'' = \mathsf{p}(\beta) \cap \mathcal{V}_{j}$, and $w = \rho^{-1}(j)$;
  \item\label{it:fact4-2} $\forall\;i \leq j \leq \ell$, $\lambda(\eta'') = 1 - \lambda(\beta'')$, with $\eta'' = \mathsf{p}(\eta) \cap \mathcal{V}_{j}$ and $\beta'' = \mathsf{p}(\beta) \cap \mathcal{V}_{j}$.
  \end{enumerate}
\end{lemma}

\begin{proof}
Let $\transpose{a_v^0}x=a_{v0}^0$, $\transpose{a_v^1}x=a_{v0}^1$ be the equations of the hyperplanes $H_{\eta}, H_{\beta}$ defined respectively by $x(\eta)[U_v]$ and $x(\beta)[U_v]$, with the normals oriented as explained on page~\pageref{nota:reflection}. We prove by induction on $\ell - i$ that the following assumption is equivalent to~\ref{it:fact4-1} and~\ref{it:fact4-2}:
\begin{enumerate}[label=\textnormal{(\emph{\roman*})},start=3,leftmargin=*,widest=iii]
  \item\label{it:fact4-3} for all $i \leq j \leq \ell$, $x(\beta'')_w=R^u x(\eta'')_w$ and $a_u \cdot a_w^0 = a_u \cdot a_w^1$, where $\eta'' = \mathsf{p}(\eta) \cap \mathcal{V}_{j}$, $\beta'' = \mathsf{p}(\beta) \cap \mathcal{V}_{j}$, $w = \rho^{-1}(j)$, and $a_w^0$ and $a_w^1$ are the normal vectors of the hyperplanes $H_{\eta''}$ and $H_{\beta''}$ oriented as usual.
\end{enumerate}

If $\ell = i$, then~\ref{it:fact4-1},~\ref{it:fact4-2}, and~\ref{it:fact4-3} hold simultaneously. Indeed, $\eta = \eta'$ and $\beta = \beta'$, hence $x(\beta)_v=R^u x(\eta)_v$ (Lemma~\ref{reflection}) and $\lambda(\eta) = 1 - \lambda(\beta)$ (Alg.~\ref{alg1}, Steps~\ref{chiinj1a} and~\ref{chiinj1b}).
In addition, we have $H_{\eta} = R^u H_{\beta}$, therefore $|a_u \cdot a_v^0| = |a_u \cdot a_v^1|$. Because the orientation of $a_v^0, a_v^1$ is such that $a_u \cdot a_v^0, a_u \cdot a_v^1 \geq 0$, the result holds.

Assume that the equivalence stated above holds for level $\ell-1$, we show that it is still the case at level $\ell$. In the sequel, denote $t = \rho^{-1}(\ell-1)$.

\smallskip
\noindent \ref{it:fact4-1} $\Leftrightarrow$ \ref{it:fact4-2}.
Suppose for all $i \leq j < \ell$, $x(\beta'')_w=R^u x(\eta'')_w$ and $\lambda(\eta'') = 1 - \lambda(\beta'')$ (by the induction hypothesis, both statements are equivalent). Hence, $H_{\eta''} = R^u H_{\beta''}$ holds for all $j$, because the $K$ points generating the hyperplanes either belong to $H_{\alpha}$, or are reflections of each other. This is true in particular if we choose $\eta'', \beta'' \in \mathcal{V}_{\ell-1}$. In addition, if we use the induction hypothesis \ref{it:fact4-1} $\Rightarrow$ \ref{it:fact4-3}), we have $a_u \cdot a_t^0 = a_u \cdot a_t^1$, so $a_t^0, a_t^1$ are directed similarly w.r.t $a_u$, and $\lambda(\eta) = 1 - \lambda(\beta)$ if and only if $x(\beta)_v=R^u x(\eta)_v$ (see Fig.~\ref{mainfig2}).
\begin{figure}[!ht]
\subfloat[$\transpose{a_v^0} x(\eta)_v > a_{v0}^0$ and $\transpose{a_v^1} x(\beta)_v > a_{v0}^1$]{\label{fig:reflection1b}
\psfrag{alpha}{$\natop{\eta}{\lambda=0}$}
\psfrag{beta}{$\natop{\beta}{\lambda=0}$}
\psfrag{ell}{$x(\beta)_v$}
\psfrag{yl}{$x(\eta)_v$}
\psfrag{refl}{\tiny $\neg$ reflection}
\psfrag{Ha}{$H_{\alpha}$}
\psfrag{He2}{$H_{\eta''}$}
\psfrag{Hb2}{$H_{\beta''}$}
\psfrag{e2}{$\eta''$}
\psfrag{b2}{$\beta''$}
\psfrag{au}{$a_u$}
\psfrag{av1}{$a_v^1$}
\psfrag{av0}{$a_v^0$}
\includegraphics[width=6cm]{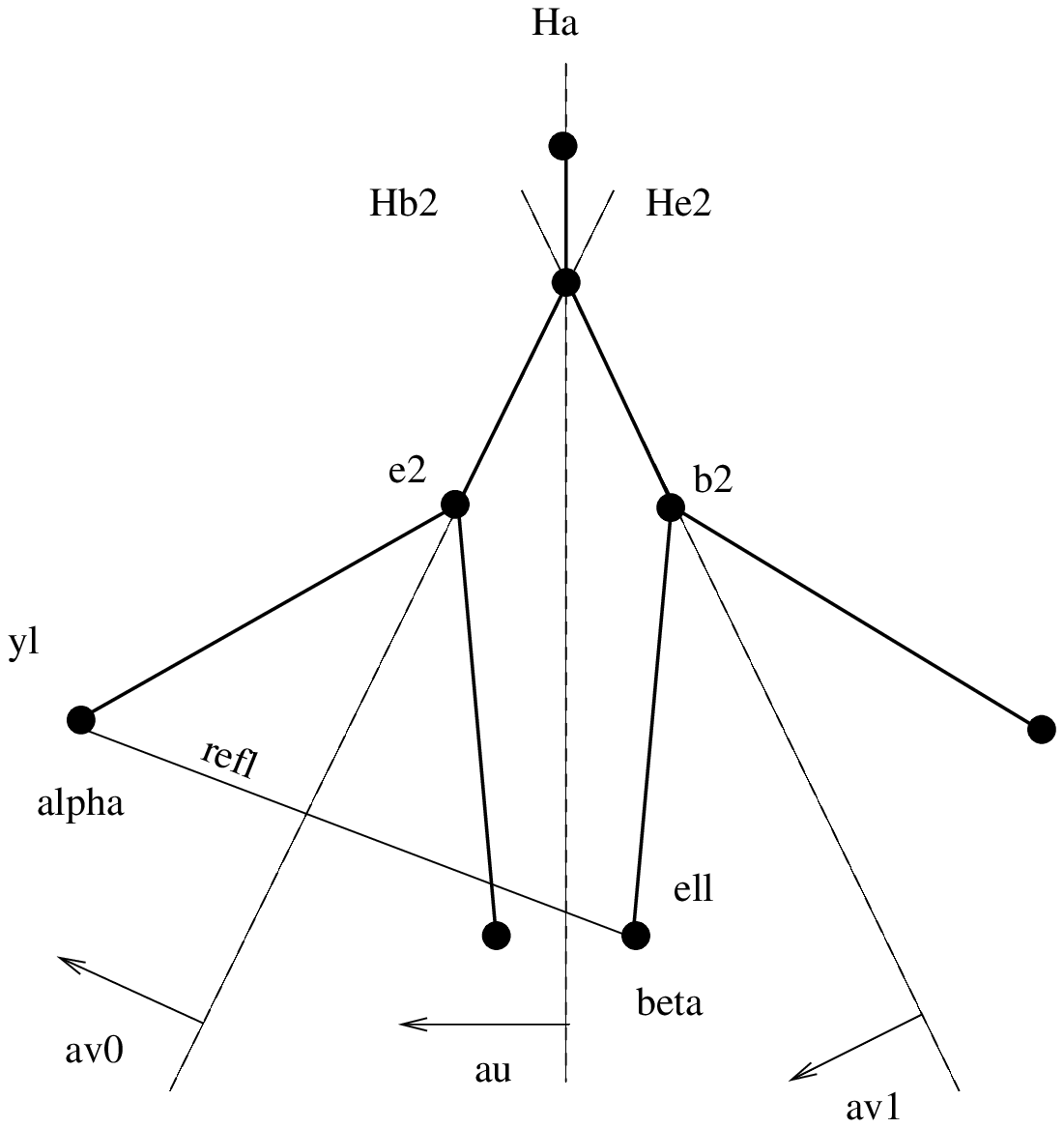}}
\hfill
\subfloat[$\transpose{a_v^0} x(\eta)_v > a_{v0}^0$ and $\transpose{a_v^1} x(\beta)_v < a_{v0}^1$]{\label{fig:reflection2b}
\psfrag{alpha}{$\natop{\eta}{\lambda=0}$}
\psfrag{beta}{$\natop{\beta}{\lambda=1}$}
\psfrag{ell}{$x(\beta)_v$}
\psfrag{yl}{$x(\eta)_v$}
\psfrag{reflection}{\tiny reflection}
\psfrag{Ha}{$H_{\alpha}$}
\psfrag{He2}{$H_{\eta''}$}
\psfrag{Hb2}{$H_{\beta''}$}
\psfrag{e2}{$\eta''$}
\psfrag{b2}{$\beta''$}
\psfrag{au}{$a_u$}
\psfrag{av1}{$a_v^1$}
\psfrag{av0}{$a_v^0$}
\includegraphics[width=6cm]{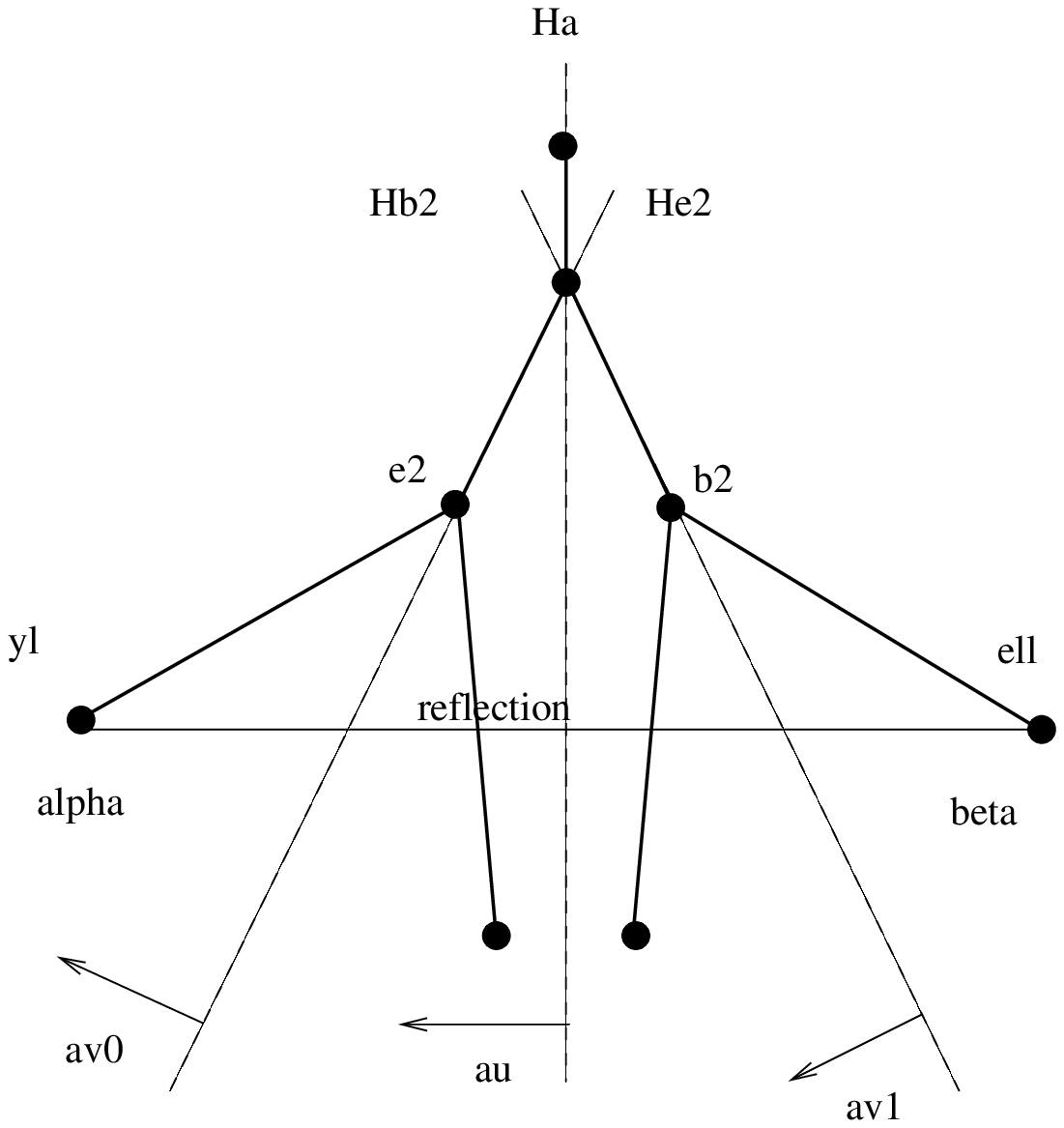}}
\caption{Proof of Lemma \ref{fact4}: Case \eqref{fig:reflection1b}
  shows the contradiction deriving from
  $\lambda(\eta)=\lambda(\beta)=0$ (or $x(\beta)_v \neq R^u
  x(\eta)_v$), and case \eqref{fig:reflection2b} the situation that
  actually occurs.}
\label{mainfig2}
\end{figure}

\smallskip
\noindent \ref{it:fact4-2} $\Rightarrow$ \ref{it:fact4-3}.
Suppose for all $i \leq j \leq \ell$, $\lambda(\eta'') = 1 - \lambda(\beta'')$. By the previous result, we also know that $i \leq j \leq \ell$, $x(\beta'')_w=R^u x(\eta'')_w$.
It remains to prove that $a_u \cdot a_v^0 = a_u \cdot a_v^1$, i.e.~that the angles $\theta^0_v$ and $\theta^1_v$ formed by these vectors have the same cosine.
Notice once again that $H_{\eta} = R^u H_{\beta}$. 
By induction, we know that the angles $\theta^0_t$, $\theta^1_t$ formed by $a_u$ and respectively $a_t^0$, $a_t^1$, have same cosine. With probability~$1$, the hyperplanes $H_{\eta}$, $H_{\beta}$ are not parallel, hence their normal vectors cannot be identical, therefore, $\theta^0_t = -\theta^1_t$ (see the illustration on Fig.~\ref{fig:angles}).
\begin{figure}[!ht]
  \begin{center}
  \begin{tikzpicture}[>=stealth]
  \draw[dashed] (0,-2) --++(0,4) node[above] {$H_{\alpha}$};
  \draw[thick,->] (0,0) -- ++(0:2) node[right] {$a_u$};

  \draw[thick,->] (0,0) -- ++(60:2) node[above right] {$a_t^0$};
  \draw[->] (0,0)+(0:0.5) arc (0:60:0.5);
  \node at (30:0.75) {$\theta^0_t$};
  \draw[thick,->] (0,0) -- ++(140:2) node[above left] {$a_v^0$};
  \draw[->] (0,0)+(60:0.5) arc (60:140:0.5);
  \node at (100:0.75) {$\theta^0$};
  \draw[->] (0,0)+(0:1.25) arc (0:140:1.25);
  \node at (70:1.5) {$\theta^0_v$};

  \draw[thick,->] (0,0) -- ++(-60:2) node[below right] {$a_t^1$};
  \draw[->] (0,0)+(0:0.5) arc (0:-60:0.5);
  \node at (-30:0.75) {$\theta^1_t$};
  \draw[thick,->] (0,0) -- ++(-140:2) node[below left] {$a_v^1$};
  \draw[->] (0,0)+(-60:0.5) arc (-60:-140:0.5);
  \node at (-100:0.75) {$\theta^1$};
  \draw[->] (0,0)+(0:1.25) arc (0:-140:1.25);
  \node at (-70:1.5) {$\theta^1_v$};
  \end{tikzpicture}
  \caption{Proof of Lemma \ref{fact4}: illustration of the fact that $a_u \cdot a_v^0 = a_u \cdot a_v^1$.}
  \label{fig:angles}
  \end{center}
\end{figure}
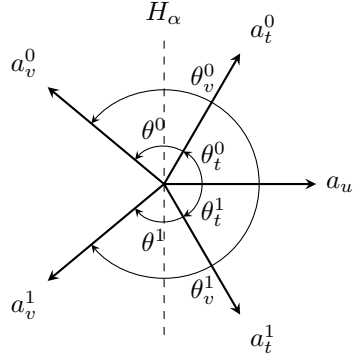
Denote $\theta^0$, $\theta^1$ the angles formed respectively by
$a_t^0$ and $a_v^0$, and by $a_t^1$ and $a_v^1$.  We also have,
$H_{\eta''} = R^u H_{\beta''}$, where $\eta'', \beta'' \in
\mathcal{V}_{\ell - 1}$, hence the normal vectors of these 4
hyperplanes are also symmetric, which implies $\theta^0 = -\theta^1$
or $\theta^0 = \pi - \theta^1$.  By the definition of $a_v^0$ and
$a_v^1$ (page~\pageref{nota:reflection}), since the scalar products
are positive, $-\pi/2 \leq \theta^0, \theta^1 \leq \pi/2$, thus
$\theta^0 = -\theta^1$.  Therefore, $\theta^0_v = \theta^0_t +
\theta^0 = -\theta^1_t - \theta^1 = -\theta^1_v$, which concludes this
part of the proof.
\smallskip
\noindent \ref{it:fact4-3} $\Rightarrow$ \ref{it:fact4-1}.
Obvious. \qed
\end{proof}

\begin{proposition}
Consider a subtree $\mathcal{T}'$ of $\mathcal{T}$ consisting of $K+2$
consecutive levels $i-K-1,\ldots,i$ (where $i\ge 2K+1$), rooted at a
single node $\eta$ and such that all nodes at all levels are marked
$\boxplus$. Let $p=2^{K+1}$ and consider the set $Y'=\{y_j\;|\;j\le
p\}$ of partial embeddings of $G$ at the leaf nodes
$\{\alpha_j\;|\;j\le p\}$ of $\mathcal{T}'$. Let $u=\rho^{-1}(i-K-1)$
and $v=\rho^{-1}(i)$. Then with probability 1 there are two distinct
positive reals $r,r'$ such that
$\|y_j(\alpha_j)_u-y_j(\alpha_j)_v\|\in\{r,r'\}$ for all $j\le p$.
\label{tworadii}
\end{proposition}
\begin{proof}
Fig.~\ref{f:2r} shows a graphical proof sketch for $K=2$.
\begin{figure}[!ht]
\begin{center}
\psfrag{e}{$\eta$}
\psfrag{b}{$\beta_1$}
\psfrag{b'}{$\beta_2$}
\psfrag{t1}{$\theta_1$}
\psfrag{t2}{$\theta_2$}
\psfrag{t1'}{$\theta_3$}
\psfrag{t2'}{$\theta_4$}
\psfrag{a1}{$\alpha_1$}
\psfrag{a2}{$\alpha_2$}
\psfrag{a3}{$\alpha_3$}
\psfrag{a4}{$\alpha_4$}
\psfrag{a1'}{$\alpha_5$}
\psfrag{a2'}{$\alpha_6$}
\psfrag{a3'}{$\alpha_7$}
\psfrag{a4'}{$\alpha_8$}
\psfrag{Rt}{$R^t$}
\psfrag{Rw}{$R^w$}
\psfrag{Rv}{$R^v$}
\psfrag{u}{$u$}
\psfrag{w}{$w$}
\psfrag{t}{$t$}
\psfrag{v}{$v$}
\psfrag{r}{$r$}
\psfrag{r'}{$r'$}
\includegraphics[width=11cm]{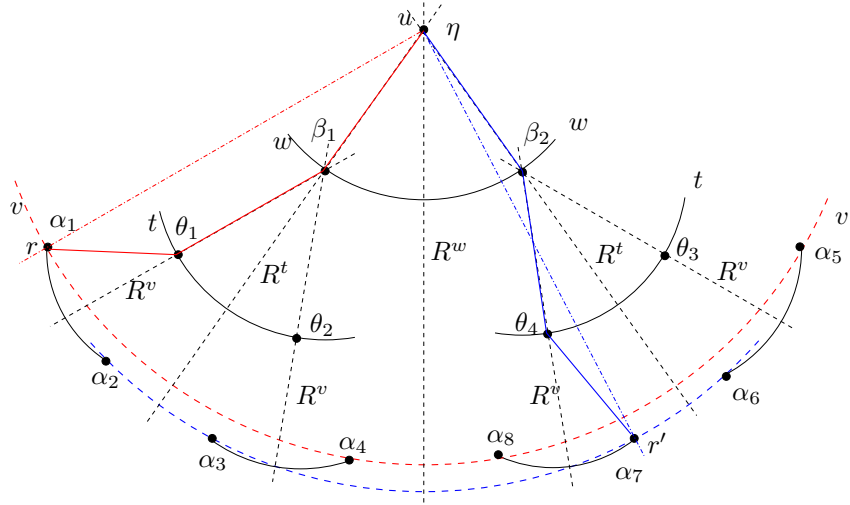}
\end{center}
\caption{Proof of Prop.~\ref{tworadii} in $\mathbb{R}^2$. The
  arrangement of three segments gives rise, in general, to two
  distances $r,r'$ between root and leaves.}
\label{f:2r}
\end{figure}  
With a slight abuse of notation, for a vertex $w\in V$ in this proof
we denote by $R^{w}$ the set of all reflections at level $w$. We order
the $\alpha_j$ nodes so that the action of $R^v$ on
$(\alpha_1,\ldots,\alpha_p)$ is the permutation $\prod_{j\!\!\mod
  2=1}(j,j+1)$. Let $t=\rho^{-1}(i-1)$.  Since all nodes are feasible,
$\|y_j(\alpha_j)_v-y_j(\alpha_j)_t\|=d_{tv}$ and
$\|y_j(\alpha_j)_u-y_j(\alpha_j)_t\|=d_{ut}$ for all $j\le p$ (we
remark that $\{t,v\}$ and $\{u,t\}$ must be in $E$ by the definition
of the GDMDGP). With probability 1, the segments through
$y_j(\alpha_j)_u$ and $y_j(\alpha_j)_t$ (where $j\le p$) do not
respectively lie within the hyperplanes defining the reflections
$R^v$; and the same holds for the segments through $y_j(\alpha_j)_t$
and $y_j(\alpha_j)_v$. Thus, there is a set $Q$ of positive reals
$r_1,\ldots,r_p$ s.t.~for all $j\le p$ with $j\mod 2=1$ we have
$\|y_j(\alpha_j)_u-y_i(\alpha_j)_v\|=r_j$ and
$\|y_{j+1}(\alpha_{j+1})_u-y_{j+1}(\alpha_{j+1})_v\|=r_{j+1}$, which
shows $|Q|\le p=2^{K+1}$. By Lemma \ref{fact4} the action of $R^t$ on
$(\alpha_1,\ldots,\alpha_p)$ is the permutation $\prod_{j\!\!\mod
  4=1}(j,j+3)(j+1,j+2)$: this implies that $r_j=r_{j+3}$ and
$r_{j+1}=r_{j+2}$ for all $j\mod 4=1$, which shows $|Q|\le
p/2=2^{K}$. Inductively, for a vertex $w$ s.t.~$i-K\le \rho(w)\le i-1$
the action of $R^w$ is $\prod_{j\!\!\mod
  2^{j-\rho(w)+1}}(j,j+2^{i-\rho(w)+1}-1)(j+1,j+2^{i-\rho(w)+1}-2)\cdots(j+2^{i-\rho(w)}-1,j+2^{i-\rho(w)})$,
which implies that $|Q|\le 2^{K+1-i+\rho(w)}$. Therefore $\rho(w)=i-K$
proves that $|Q|\le 2$. The case $|Q|=1$ can only occur if
$y_j(\alpha_j)_u$, $y_j(\alpha_j)_t$ and $y_j(\alpha_j)_v$ are
collinear for all $j\le p$, an event that occurs with probability
0. \qed
\end{proof}

Prop.~\ref{tworadii} is useful in order to show that certain
configurations of nodes within $\mathcal{T}$ can only occur with
probability 0.
\begin{example}
Consider a subtree $\mathcal{T}'$ of $\mathcal{T}$ like the one in
Fig.~\ref{f:2r} embedded in $\mathbb{R}^2$, and suppose that all nodes
at level $u,w,t$ are marked $\boxplus$, and further that only one node
within $\alpha_1,\alpha_2$ is feasible, only one node within
$\alpha_3,\alpha_4$ is feasible, only one node within
$\alpha_7,\alpha_8$ is feasible, and $\alpha_5,\alpha_6$ are both
infeasible. This must be due to a distance $d_{u'v}$ with $u'\le
u$. Consider now a circle $C$ completely determined by its center at
$y_1(\alpha_1)_{u'}$ and its radius $d_{u'v}$; if $C$ also contains
the points at the nodes $\alpha_1,\alpha_4,\alpha_8$ or the points at
the nodes $\alpha_2,\alpha_3,\alpha_7$ then we must have $u'=u$, in
which case also one of $\alpha_5,\alpha_6$ will be feasible (against
the hypothesis). And the probability that $C$ should contain the
points at the nodes $\alpha_1,\alpha_3,\alpha_8$ or
$\alpha_2,\alpha_4,\alpha_7$ is zero. Hence $\mathcal{T}'$ can only
occur with probability 0. \qed
\end{example}

\begin{corollary}
Consider a subtree $\mathcal{T}'$ of $\mathcal{T}$ consisting of
$K+q+1$ consecutive levels $i-K-q,\ldots,i$ (where $i\ge 2K+q$ and
$q\ge 1$), rooted at a single node $\eta$ and such that all nodes at
all levels are marked $\boxplus$. Let $p=2^{K+q}$ and consider the set
$Y'=\{y_j\;|\;j\le p\}$ of partial embeddings of $G$ at the leaf nodes
$\{\alpha_j\;|\;j\le p\}$ of $\mathcal{T}'$. Let $u=\rho^{-1}(i-K-q)$
and $v=\rho^{-1}(i)$. Then with probability 1 there is a set
$H^{uv}=\{r_j\;|\;j\le 2^q\}$ of $2^q$ distinct positive reals such
that $\|y_i(\alpha_i)_u-y_i(\alpha_i)_v\|\in H^{uv}$ for all $i\le p$.
\label{2r1}
\end{corollary}
\begin{proof}
The proof of Prop.~\ref{tworadii} can be generalized to span an
arbitrary number of levels by induction on $q$. Two distances
$r_{j_1},r_{j_2}\in H^{uv}$ can only be equal by collinearity of some
subsets of points, an event occurring with probability 0. \qed
\end{proof}

\begin{corollary}
Let $y\in X$ and $v\in V\smallsetminus V_0$ such that $\Upsilon(y,v)$
holds. If $\{u,w\}\in E$ with $u<v<w$ and $\rho(w)-\rho(u)>K$ then
$d_{uw}\in H^{uw}$ with probability 1. \label{2r2}
\end{corollary}
\begin{proof}
Since $\Upsilon(y,v)$ holds, then the GDMDGP instance is YES and there
must exist at least two feasible nodes at level $\rho(w)$ in
$\mathcal{T}$. If $d_{uw}\not\in H^{uw}$ the probability that a
completely determined sphere contains two arbitrary points in
$\mathbb{R}^K$ is zero. Since the instance is a YES one, however, the
BP algorithm does not prune all feasible nodes due to $d_{uw}$. By
Cor.~\ref{2r1} the only remaining possibility (which therefore
occurs with probability 1) is that $d_{uw}\in H^{uw}$. \qed
\end{proof}

\begin{corollary}
Let $y\in X$ and $v\in V\smallsetminus V_0$ such that $\Upsilon(y,v)$
holds. If $u\in V$ with $u>v$ then $R^v y_u$ belongs to a valid
extension of $y[U_v]$. \label{2r3}
\end{corollary}
\begin{proof}
If there is no edge $\{w,u\}\in E$ with $\rho(u)-\rho(w)>K$ the result
follows by Cor.~\ref{reflectionbranch}. Otherwise, by Cor.~\ref{2r2},
$d_{wu}\in H^{wu}$. As in the proof of Prop.~\ref{tworadii}, all pairs of
points that are feasible w.r.t.~$d_{wu}$ are reflections of each other
w.r.t.~$R^v$. \qed
\end{proof}

\begin{theorem}
Let $y\in X$ and $v\in V\smallsetminus V_0$ such that $\Upsilon(y,v)$
holds. Then $\tilde{R}^vy\in X$ with probability 1. \label{fact5}
\end{theorem}
\begin{proof}
We have to show that $\tilde{R}^vy$ is a valid embedding for
$G=(V,E)$. Partition $E$ into three subsets $E_1,E_2,E_3$, where
$E_1=\{\{t,u\}\in E\;|\;t,u<v\}$, $E_2=\{\{t,u\}\in E\;|\;t,u\ge v\}$
and $E_3=\{\{t,u\}\in E\;|\;t<v\land u\ge v\}$. For $E_1$, by
definition
$\|(\tilde{R}^vy)_t-(\tilde{R}^vy)_u)\|=\|Iy_t-Iy_u\|=\|y_t-y_u\|=d_{tu}$
as claimed. For $E_2$,
$\|(\tilde{R}^vy)_t-(\tilde{R}^vy)_u)\|=\|R^vy_t-R^vy_u\|=\|y_t-y_u\|=d_{tu}$
because $R^v$ is an isometry. For $E_3$, we aim to show that
$\|Iy_t-R^vy_u\|=d_{tu}$. Since $y\in X$, by Lemma \ref{fact1} there
is a feasible leaf node $\alpha$ with $x(\alpha)=y$. Because
$\Upsilon(y,v)$, $\exists\eta\in\mathcal{V}_{\rho(v)-1}$ such that
$x(\eta)=y[\gamma(v)]$ and $N^+(\eta)=\{\beta,\beta'\}$ with
$\mu(\beta)=\mu(\beta')=\boxplus$; we can assume without loss of
generality that $\mbox{\sf
  p}(\alpha)\cap\mathcal{V}_{\rho(v)}=\{\beta\}$; furthermore, again
by $\Upsilon(y,v)$, there is at least one feasible leaf node $\alpha'$
such that $\mbox{\sf
  p}(\alpha')\cap\mathcal{V}_{\rho(v)}=\{\beta'\}$. Let
$\{\omega\}=\mbox{\sf p}(\alpha)\cap\mathcal{V}_{\rho(u)}$ and
$\{\omega'\}=\mbox{\sf p}(\alpha')\cap\mathcal{V}_{\rho(u)}$.  Because
$\omega'$ is feasible, $\|x(\omega')_t-x(\omega')_u\|=d_{tu}$; because
$\eta$ is an ancestor of both $\alpha$ and $\alpha'$ at level
$\rho(v)-1$ and $t<v$, $\mbox{\sf
  p}(\alpha')\cap\mathcal{V}_{\rho(t)}=\mbox{\sf
  p}(\alpha)\cap\mathcal{V}_{\rho(t)}$, which implies that
$x(\omega')_t=x(\omega)_t=y_t$. Thus,
$\|y_t-y_u\|=d_{tu}=\|y_t-x(\omega')_u\|$. Furthermore, because
$\beta'\in\mbox{\sf p}(\omega')\cap\mathcal{V}_{\rho(v)}$,
$x(\omega')$ extends $x(\beta')$. By Alg.~\ref{alg1},
Steps~\ref{chiinj1a} and~\ref{chiinj1b}, $\lambda(\beta)=1-\lambda(\beta')$. Because
$\alpha$ is feasible, at every level $\rho(u')\in V$ such that $v\le
u'<u$ the node $\theta\in\mbox{\sf
  p}(\alpha)\cap\mathcal{V}_{\rho(u')}$ has $f\in\{1,2\}$ feasible
subnodes; by Prop.~\ref{uniformlevel}, the node $\theta'\in\mbox{\sf
  p}(\alpha')\cap\mathcal{V}_{\rho(u')}$ also has $f$ feasible
subnodes. If $f=2$, by Cor.~\ref{2r3} it is possible to choose
$\alpha'$ so that $\lambda(\theta')=1-\lambda(\theta)$ with
probability 1; if $f=1$ then by Alg.~\ref{alg1}, Steps~\ref{chiinj2a} and~\ref{chiinj2b},
all feasible nodes inherit the same $\lambda$ value as their parents,
so $\lambda(\theta')=1-\lambda(\theta)$. By Lemma \ref{fact4},
$x(\omega')_u=R^vy_u$ with probability 1. Hence
$\|y_t-R^vy_u\|=d_{tu}$ as claimed.  \qed
\end{proof}

\section{Symmetry and Number of Solutions}\label{s:nbsol}

Our strategy for proving that feasible GDMDGP instances have power of
two solutions with probability 1 is as follows. We map embeddings
$y\in X$ to binary sequences $\chi\in\{0,1\}^n$ describing the
``branching path'' in the tree $\mathcal{T}$. We define a symmetry operation on
$\chi$ by flipping its tail from a given component $i$ to its end
(this operation is akin to branching at level $i$). We show that the cardinality of the group of all such symmetries is a power of two by
bijection with a set of binary sequences. Finally we prove that the
cardinality of the symmetry group is the same as $|X|$.

For all leaf nodes $\alpha\in\mathcal{V}$ with $\mu(\alpha)=\boxplus$
let $\chi(\alpha)=(\lambda(\beta)\;|\;\beta\in\mbox{\sf p}(\alpha))$;
since embeddings in $X$ are also in correspondence with leaf
$\boxplus$-nodes of $\mathcal{T}$ by Alg.~\ref{alg1}, Step
\ref{valid1}, $\chi$ defines a relation on $X\times \{0,1\}^n$.
\begin{lemma}
With probability 1, the relation $\chi$ is a function.
\label{chiwd}
\end{lemma}
\begin{proof}
For $\chi$ to fail to be well-defined, there must exist an embedding
$x$ which is in relation with two distinct binary sequences
$\chi',\chi''$, which corresponds to the discriminant of the quadratic
equation in the proof of Lemma \ref{atmosttwolemma} taking value zero
at some rank $>K$, which happens with probability 0. \qed
\end{proof}
Let $\Xi=\{\chi(y)\;|\;y\in X\}$. For $y\in X$ let $y^i$ be its
subsequence $(x_1,\ldots,x_i)$. We extend $\chi$ to be defined on all
such subsequences by simply setting
$\chi^i=(\chi(y)_1,\ldots,\chi(y)_i)$; $\chi(y)$ is valid if $y$ is a
valid embedding.

Let $N=\{1,\ldots,n\}$ and $g$ be the $n\times n$ binary matrix such
that $g_{ij} = 1$ if $i \le j$ and $0$ otherwise (the upper triangular
$n\times n$ all-1 matrix); let $g_i$ be its $i$-th row vector and
$\Gamma=\{g_i\;|\; i\in N\}$. Consider the elementwise modulo-2
addition in the set $\mathbb{F}_2^n$ (denoted $\oplus$): this endows
$\mathbb{F}_2^n$ with an additive group structure with identity
$e=(0,\ldots,0)$ where each element is idempotent. Thus,
$\mathcal{G}=(\mathbb{F}_2^n,\oplus)\cong C_2^n$. This group naturally
acts on itself (and subsets thereof) using the same $\oplus$
operation.  It is not difficult to prove that $\Gamma$ is a set of
group generators for $\mathcal{G}$ and a linearly independent set of
the vector space $\mathcal{V}$ given by $\mathcal{G}$ with scalar
multiplication over $\mathbb{F}_2$. For all $S\subseteq N$, let
\begin{equation*}
  g_S=\bigoplus\limits_{i\in S}g_i,
\end{equation*}
and define a mapping $\phi:\mathcal{P}(N)\to\mathcal{G}$ given by
$\phi(S)=g_S$.
\begin{lemma}
  $\phi$ is injective.
  \label{injthm}
\end{lemma}
\begin{proof}
We show that for all $S,T\subseteq N$, if $g_S=g_T$ then $S=T$.
\begin{equation*}
\begin{array}{llrcl} 
  &&g_S &=& g_T \\ [0.3em]
  &\Rightarrow & \bigoplus\limits_{i\in S}g_i &=&
    \bigoplus\limits_{i\in T} g_i \\ [0.3em]
  &\Rightarrow & \bigoplus\limits_{i\in S} g_i \oplus
    \bigoplus\limits_{i\in T} g_i^{-1} &=& e\\ [0.3em]
  \mbox{idempotency} &\Rightarrow & \bigoplus\limits_{i\in S} g_i \oplus
    \bigoplus\limits_{i\in T} g_i &=& e\\ [0.3em]
  \mbox{$g_i\oplus g_i=g_i^2$}  &\Rightarrow &
    \bigoplus\limits_{i\in S\triangle T} g_i \oplus
    \bigoplus\limits_{i\in S\cap T} g_i^2 &=& e \\ [0.3em] 
  \mbox{idempotency}& \Rightarrow &\bigoplus\limits_{i\in S\triangle
    T} g_i &=&e\\ [0.3em]
  \mbox{linear independence} &\Rightarrow &S\triangle T &=& \emptyset
    \\ [0.3em]
  &\Rightarrow & S &=& T.
\end{array}
\end{equation*}
This concludes the proof. \qed
\end{proof}

\begin{lemma}
 For all $H\subseteq\Gamma$, $|\langle H\rangle|=2^{|H|}$.
 \label{cardcor}
\end{lemma}
\begin{proof}
The restriction of function $\phi$ to $\mathcal{P}(H)$ is injective by
Lemma~\ref{injthm}. Furthermore, each element $g$ of $\langle H\rangle$
can be written as $\bigoplus\limits_{i\in S} g_i$ for some $S\subseteq H$
because $H$ is a spanning set for the vector space $H$ over
$\mathbb{F}^n_2$, which is setwise equal to the group $\langle
H\rangle$. Thus $\phi$ is surjective too. Hence $\phi$ is a bijection
between $\mathcal{P}(H)$ and $\langle H\rangle$, which yields the
result. \qed
\end{proof}

Let $I$ be the set of levels of $\mathcal{T}$ for which from all nodes
with two valid children there is a path going to a feasible leaf
through both children. Let $L=\{g_i\in\Gamma\;|\;i\in I\}$ and
$\Lambda=\langle L\rangle$ be the subgroup of $\mathcal{G}$ generated
by $L$.

\begin{theorem}
If $\Xi\not=\emptyset$, for all $\xi\in \Xi$ we have $\xi\oplus\Lambda=\Xi$
with probability 1.
\label{mainthm}
\end{theorem}
\begin{proof}
($\Rightarrow$) We show that $\xi\oplus\Lambda\subseteq \Xi$ with
  probability 1; because $\langle L\rangle=\Lambda$ it suffices to
  show that $\xi\oplus g_i\in \Xi$ for an arbitrary $g_i\in L$,
  i.e.~that there exists a valid embedding $w\in X$ such that
  $\chi(w)=\xi\oplus g_i$. Let $y\in\chi^{-1}(\xi)$ and
  $v=\rho^{-1}(i)$ such that $\Upsilon(y,v)$, and define
  $w=\tilde{R}^vy$ (where $\tilde{R}^v$ is defined in Thm.~\ref{fact5}
  above); by Thm.~\ref{fact5}, $w\in X$. Let $\alpha'$ be the leaf
  node of $\mathcal{T}$ such that $x(\alpha')=y$; by
  Lemma~\ref{fact1}, there is a leaf node $\beta'$ such that
  $x(\beta')=w$. We have to show that for all $\ell\ge i$ the node
  $\beta\in\mbox{\sf p}(\beta')\cap\mathcal{V}_\ell$ is such that
  $\lambda(\beta)=1-\lambda(\alpha)$, where $\alpha$ is the node in
  $\mbox{\sf p}(\alpha')\cap\mathcal{V}_\ell$. We proceed by induction
  on $\ell$. For $\ell=i$ this holds by Lemma \ref{reflection}. For
  $\ell>i$, the induction hypothesis allows us to apply Lemma
  \ref{fact4} and conclude that the event
  $\lambda(\alpha)=1-\lambda(\beta)$ occurs with probability 1.

($\Leftarrow$) Now we show that $\Xi\subseteq\xi\oplus\Lambda$ with
  probability 1, i.e.~for any $\eta\in \Xi$ there is $g\in\Lambda$ with
  $\xi\oplus g=\eta$. We proceed by induction on $n$, which starts
  when $n=K+1$: if $K+1\not\in I$ then $|\Xi|=1$, $L=\emptyset$ and the
  theorem holds; if $K+1\in I$ then $|\Xi|=2$, $L=\{g_{K+1}\}$ and
  the theorem holds. Now let $n>K+1$; for all $j\in\{K+1,\ldots,n-1\}$
  define $\Xi^j=\{\xi^j\;|\;\xi\in \Xi\}$ and
  $L^j=\{g_\ell\in\Gamma\;|\;\ell\in I\land \ell\le j\}$. By the
  induction hypothesis, for all $\xi'\in \Xi^j\;(\xi'\oplus\langle
  L^j\rangle=\Xi^j)$. Now, either $n\not\in I$ or $n\in I$; by
  Prop.~\ref{uniformlevel}, with probability 1 if $n\not\in I$ then nodes
  in $\mathcal{V}_{n-1}$ can only have zero or one feasible subnode
  (let $B_1^n$ be the set of all such feasible subnodes), and if
  $n\in I$ then nodes in $\mathcal{V}_{n-1}$ can only have zero or
  two feasible subnodes $\beta$ (let $B_2^n$ be the set of all such
  feasible subnodes). In the former case we let
  $\Xi^n=\{\xi(x(\beta))\;|\;\beta\in B_1^n\}$ and $L^n=L^{n-1}$; in the
  latter we let $\Xi^n=\{\xi(x(\beta))\;|\;\beta\in B_2^n\}$ and
  $L^n=L^{n-1}\cup\{g_n\}$. In both cases it is easy to verify that
  the theorem holds for $\Xi^n,L^n$: in the former case it follows by
  the induction hypothesis, and in the latter case it follows because
  $g_n=(0,\ldots,0,1)$, namely, if $\eta\in \Xi$ and $n\in I$ then
  take $\xi=\eta\oplus g_n$ (the result follows by idempotency of
  $g_n$). \qed
\end{proof}

\begin{corollary}
If a GDMDGP instance is feasible, $|X|$ is a power of two with
probability 1.
\label{poweroftwo}
\end{corollary}
\begin{proof}
By Lemma \ref{chiwd} $\chi$ is a function with probability 1.  Let
$x,x'\in X$ be distinct; then by Alg.~\ref{alg1}, Steps~\ref{chiinj1a},~\ref{chiinj1b},~\ref{chiinj2a}, and~\ref{chiinj2b}, the map $\chi:X\to \Xi$ is injective. By definition
of $\Xi$ it is also surjective, hence $|X|=|\Xi|$. By
Thm.~\ref{mainthm} $|\Xi|=|\chi\oplus\Lambda|$ for all $\chi\in \Xi$
with probability 1.  It is easy to show that
$|\chi\oplus\Lambda|=|\Lambda|$, so by Lemma \ref{cardcor} $|X|$ is a
power of two with probability 1. \qed
\end{proof}

\section{Counterexamples}
\label{s:counterex}
We discuss a class of counterexamples to the conjecture that {\it all}
GDMDGP instances have a number of solutions which is a power of two
(also see Lemma 5.1 in \cite{lln2}). All these counterexamples are
hand-crafted and have the property that two distinct embeddings $x,x'$
have at least a level $i$ where $x_i=x'_i$, which is an event which
happens with probability 0. For any $K\ge 1$, let $n=K+3$,
$V=\{1,\ldots,n\}$, $E=\{\{i,j\}\;|\;0<i-j\le K\}\cup\{\{1,n\}\}$ and
$d_{ij}=1$ for all $\{i,j\}\in E$. The first $n-2=K+1$ points can be
embedded in the vertices of a regular simplex in dimension $K$; then
either $x_{n-1}=x_1$ or $x_{n-1}$ is the symmetric position from $x_1$
with respect to the hyperplane through $\{x_2,\ldots,x_{n-2}\}$. In
the first case, the two positions for $x_n$ are valid, in the second
only $x_n=x_2$ is possible (see Fig.~\ref{fig:ctr-ex-h2} for the
2-dimensional case), yielding a YES instance where $|X|=6$.
\begin{figure}[!ht]
\begin{center}
\subfloat[Positions of the points on the plane.]{\label{fig:ctr-ex-h21}
  \begin{tikzpicture}[>=stealth,scale=0.25]
    \draw (0,0) -- (8,0) -- (4,7) --cycle;
    \draw[dashed] (8,0) -- (12,7) -- (4,7);
    \draw[dotted] (0,0) -- (-4,7) -- (4,7);
    \draw[dotted] (4,7) -- (8,14) -- (12,7);
    \node[anchor=base] at (-1,-2) {$x_1 = x^{(0)}_4$};
    \node[anchor=base] at (11,-2) {$x_2 = \textcolor{darkgreen}{x^{(01)}_5} = \textcolor{darkgreen}{x^{(11)}_5}$};
    \node[anchor=base] at (3.5,8) {$x_3$};
    \node[anchor=base] at (13,8) {$x^{(1)}_4$};
    \node[anchor=base] at (-4,8) {$\textcolor{darkgreen}{x^{(00)}_5}$};
    \node[anchor=base] at (8,15) {$\textcolor{darkred}{x^{(10)}_5}$};
    \path (0,-3.2);
  \end{tikzpicture}
}\hfill\subfloat[BP tree.]{\label{fig:ctr-ex-h22}
  \begin{tikzpicture}[>=stealth,scale=0.9]
    \node[anchor=base] (z) at (0,0) {$x_1$};
    \node[anchor=base] (u) at (0,-1) {$x_2$};
    \node[anchor=base] (d) at (-1,-2) {$x_3$};
    \node[anchor=base] (t0) at (-3,-3) {$x^{(0)}_4$};
    \node[anchor=base] (t1) at (-1,-3) {$x^{(1)}_4$};
    \node[anchor=base] (q00) at (-4,-4) {\color{darkgreen} ${x^{(00)}_5}$};
    \node[anchor=base] (q01) at (-3,-4) {\color{darkgreen} ${x^{(01)}_5}$};
    \node[anchor=base] (q10) at (-2,-4) {\color{darkred} ${x^{(10)}_5}$};
    \node[anchor=base] (q11) at (-1,-4) {\color{darkgreen} ${x^{(11)}_5}$};
    \draw (z) -- (u) -- (d) -- (t0) -- (q00)
                               (t0) -- (q01)
                        (d) -- (t1) -- (q10)
                               (t1) -- (q11);
    \draw[dashed] (u) -- (0,-4) -- (2,-4) --cycle;
    \node at (0.8,-3.5) {symmetric};
    \path (0,-4.3);
  \end{tikzpicture}
}
\end{center}
\caption{The counterexample in the case $K=2$. Embeddings
  $x^{(00)}_5$, $x^{(01)}_5$, and $x^{(11)}_5$ are valid, while
  $x^{(10)}_5$ is not.}
\label{fig:ctr-ex-h2}  
\end{figure}
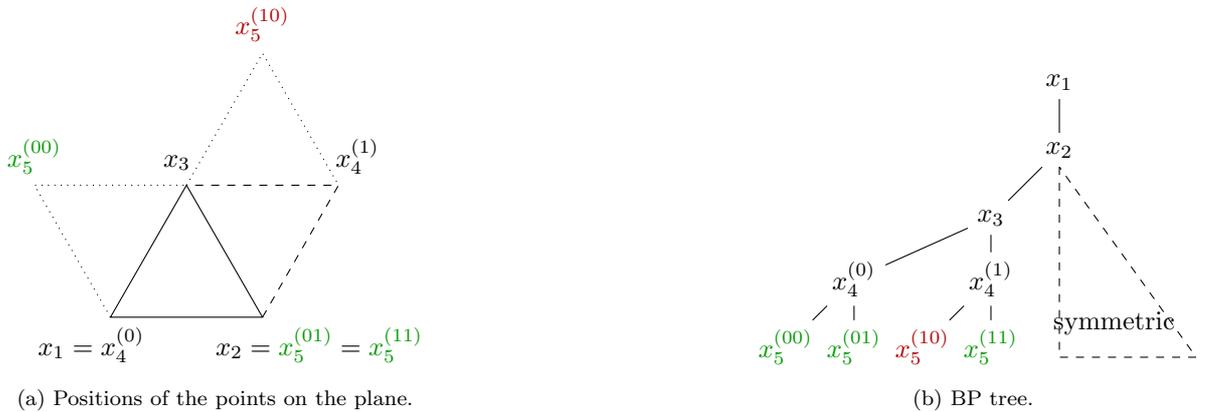

Lastly, Fig.~\ref{fig:ctr52} shows an example where the~\ref{it:fact4-2}~$\Rightarrow$~\ref{it:fact4-1} implication of Lemma \ref{fact4} fails for instances in
$\mbox{DDGP}\smallsetminus\mbox{GDMDGP}$. This shows that any
generalization of our result to the DDGP is not trivial. Let
$V=\{1,\ldots,6\}$ (the graph drawing is the same as the embedding in
$\mathbb{R}^2$). The nodes $5',6'$ linked with dashed lines show
alternative node placements. Let $U_5=\{3,4\}$ and $U_6=\{1,2\}$. The
line through the points $3,4$ does not provide a valid reflection
mapping $6$ to $6'$. This happens because $U_6$ does not consist of
the two {\it immediate} predecessors of $6$.
\begin{figure}[!ht]
\psfrag{1}{$1$}
\psfrag{2}{$2$}
\psfrag{3}{$3$}
\psfrag{4}{$4$}
\psfrag{5}{$5$}
\psfrag{5'}{$5'$}
\psfrag{6}{$6$}
\psfrag{6'}{$6'$}
\psfrag{Rv}{$R^v$}
\psfrag{U5}{$U_5$}
\psfrag{U6}{$U_6$}
\begin{center}
\includegraphics[width=10cm]{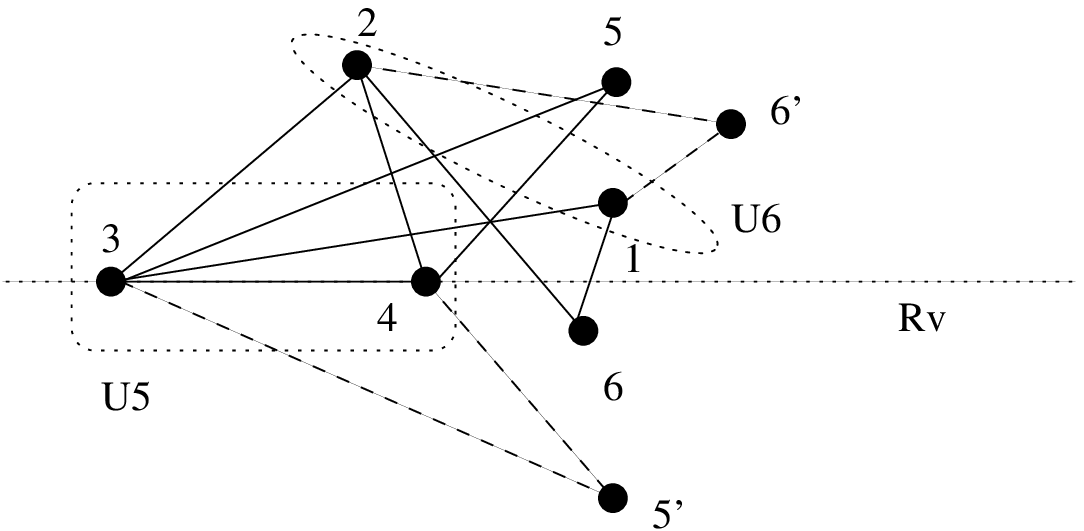}
\end{center}
\caption{A counterexample to Lemma \ref{fact4} applied to
  $\mbox{DDGP}\smallsetminus\mbox{GDMDGP}$.}
\label{fig:ctr52}
\end{figure}



\ifpaper
\bibliographystyle{splncs}
\else
\bibliographystyle{plain}
\fi
\bibliography{mdc}

\end{document}